\RequirePackage{amsmath}
\documentclass[]{llncs}
\usepackage{amsmath,amssymb,mathtools}
\usepackage[final]{microtype}
\usepackage[linesnumbered, ruled]{algorithm2e}
\usepackage[inline]{enumitem}
\usepackage[T1]{fontenc}
\usepackage{cite}

\usepackage[T1]{fontenc}

\usepackage{lipsum}

\def\R{\mathbb{R}}

\def\Problem{\textsc{Minimum Dominating Tree}}
\def\Prob{MDT}
\def\ProblemStar{\textsc{Minimum Dominating Star}}
\def\ProbStar{\textsc{MDS}}

\def\ProblemGroup{\textsc{Group Steiner Tree}}
\def\ProbGroup{\textsc{GST}}

\def\ProblemDom{\textsc{Minimum Dominating Set}}
\def\ProbDom{\textsc{DOM}}

\def\ProblemSetCover{\textsc{Minimum Set Cover}}
\def\ProbSetCover{\textsc{SC}}

\def\ProblemPath{\textsc{Minimum Dominating Path}}
\def\ProbPath{\textsc{MDP}}

\def\ProblemHam{\textsc{Hamiltonian Path}}
\def\ProbHam{\textsc{HP}}
\usepackage{tikz}
\usetikzlibrary{
	arrows,
	arrows.meta,
	calc,
	graphs,
	patterns,
	positioning,
	shapes,
	decorations.pathmorphing,
}

\tikzset{label/.style={
	draw=none,
	sloped,
}}

\tikzset{label above/.style={
	label,
	midway,
	above=-1mm,
}}

\tikzset{label below/.style={
	label,
	midway,
	below=-1mm,
}}

\tikzset{
	>=Latex
	,every picture/.style={
		very thick
	}
	,every node/.style={
		draw 
		,circle
		,inner sep=0mm
		,minimum size=5mm
		,very thick
		,font=\small
	}
}

\def\titletext{Hardness Results and Approximation Algorithms for the Minimum Dominating Tree Problem}
\title{\titletext}

\author{
Gilad Kutiel\inst{1}
}

\institute{
Department of Computer Science, Technion, Haifa, Israel
\\
\email{gkutiel@cs.technion.ac.il}
}

\date{}

\begin{document}
\maketitle

\begin{abstract}
Given an undirected graph $G = (V, E)$ and a weight function $w:E \to \R$, 
the \Problem{} problem asks to find a minimum weight sub-tree of $G$, 
$T = (U, F)$, such that every $v \in V \setminus U$ is adjacent to at least one 
vertex in $U$.
The special case when the weight function is uniform is known as the 
\textsc{Minimum Connected Dominating Set} problem.

Given an undirected graph $G = (V, E)$ with some subsets of vertices called groups,
and a weight function $w:E \to \R$,
the \textsc{Group Steiner Tree} problem is to find a minimum weight sub-tree
of $G$ which contains at least one vertex from each group. 

In this paper we show that the two problems are equivalents 
from approximability perspective.
This improves upon both the best known approximation algorithm and the best 
inapproximability result for the \Problem{} problem.
We also consider two extrema variants of the \Problem{} problem, namely,
the \ProblemStar{} and the \ProblemPath{} problems 
which ask to find a minimum dominating star and path respectively.  
\end{abstract}

\section{Introduction}
Given an undirected graph $G = (V, E)$ and a weight function $w:E \to \R$, 
the \Problem{} problem (\Prob{}) asks to find a minimum weight sub-tree of $G$, 
$T = (U, F)$, such that every $v \in V \setminus U$ is adjacent to at least one 
vertex in $U$.
The special case when the weight function is uniform is known as the 
\textsc{Minimum Connected Dominating Set} problem (CDS). 
Both CDS and \Prob{} have many applications in routing
for mobile ad-hoc networks, see for 
example~\cite{shin2010approximation,cheng2003polynomial,das1997routing,adasme2016models,adasme2017minimum}.
Figure~\ref{fig:problem} shows an example instance and a possible solution to the problem. 

\begin{figure}
\begin{center}
\begin{tikzpicture}[]
\foreach[count=\i] \x \y in {
	0/0
	,1/2,1/-1
	,2/0
	,3/1,3/-2
	,4/-1
}{
	\node(\i) at(\x,\y) {\i};
}

\foreach \u \v \w in {
	1/2,1/3%
	,2/4,2/5%
	,3/4,3/6,3/7%
	,4/5,4/7%
	,5/7%
}{
	\draw (\u) -- (\v) node[label above]{\w};
}

\begin{scope}[xshift=6cm]
\foreach[count=\i] \x \y in {
	0/0
	,1/2,1/-1
	,2/0
	,3/1,3/-2
	,4/-1
}{
	\node(\i) at(\x,\y) {\i};
}

\foreach \u \v \w in {
	1/2%
	,2/4%
	,3/4%
	,4/7%
}{
	\draw (\u) -- (\v) node[label above]{\w};
}

\foreach \u \v in {
	1/3%
	,2/5%
	,3/6,3/7%
	,4/5%
	,5/7%
}{
	\draw[thin, dashed] (\u) -- (\v);
}
\end{scope}
\end{tikzpicture}
\end{center}
\caption{\label{fig:problem}
From left to right
a) An instance of the \Problem{} problem.
b) A possible solution (solid edges) with value of 17.
}
\end{figure}
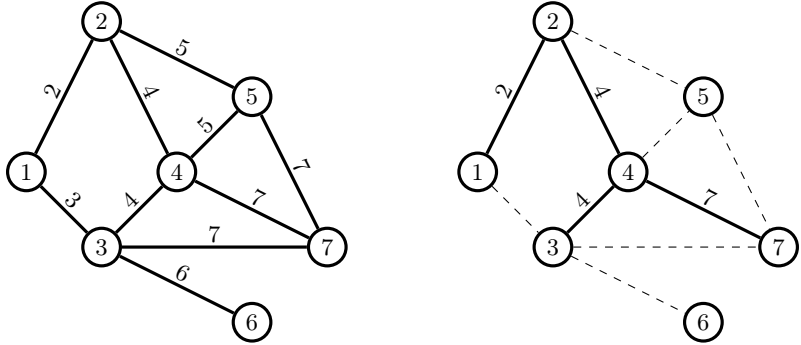

Given an undirected graph $G = (V, E)$ with some subsets of vertices called groups,
and a weight function $w:E \to \R$,
the \ProblemGroup{} problem (\ProbGroup{}) is to find a minimum weight sub-tree
of $G$ which contains at least one vertex from each group.
\ProbGroup{} is not approximable within $ \Omega(\log^{2\textbf{-}\epsilon}n)$
unless NP admits quasi-polynomial-time Las-Vegas
algorithm~\cite{halperin2003polylogarithmic}.
On the other hand, there is a $\log^3 n$-approximation algorithm for the 
problem~\cite{garg2000polylogarithmic}.

In this paper we show that the two problems are equivalents 
from approximation algorithms perspective.
This improves upon both the best known approximation algorithm and the best 
inapproximability result for \Prob{}.
We also consider two extrema variants of \Prob{}, namely,
the \ProblemStar{} (\ProbStar{}) and the \ProblemPath{} (\ProbPath{}) 
problems which ask to find a minimum dominating star and path respectively.  

\textbf{Previous Work:}
CDS has a long history starting at the late 70s~\cite{sampathkumar1979connected},
and it is approximable within $\ln\Delta + 3$~\cite{guha1998approximation} 
where $\Delta$ is the maximum degree in $G$,
which is the best one can wish for if $P \neq NP$. 

\Prob{}, to the best of our knowledge, was introduced in~\cite{shin2010approximation}.
In the same paper it was shown that the \textsc{Minimum Weighted Dominating Set} problem 
can be reduced to \Prob{} in a way that preserve the approximation ratio.
Thus, there is no $c\log n$-approximation algorithm, for some $c > 0$, for \Prob{}
unless $P = NP$.
In the same paper it was shown that \Prob{} can be reduced to the 
\textsc{Minimum Directed Steiner Tree} problem in a way that preserve the approximation ratio.
Unfortunately, the current best approximation algorithm for the 
\textsc{Minimum Directed Steiner Tree} problem yields a $|S|^\varepsilon$ 
approximation ratio~\cite{charikar1999approximation},
where $S$ is the set of terminals.
To the best of our knowledge, this is the best approximation algorithm known for \Prob{}.
The existence of a dominating path in a graph was studied in several 
papers see~\cite{broersma1988existence, faudree2017minimum} for example,
but, to the best of our knowledge, it was never considered from algorithmic perspective.

\textbf{Our Result:}
We show that the \Problem{} problem is equivalent to the \ProblemGroup{} problem
from approximability perspective, 
by doing so we prove that \Prob{} is inapproximable within $\log^{2 - \varepsilon} n$   
unless $NP \subseteq \text{ZTIME}(n^{\text{polylog}(n)})$.
This also, directly leads to a $\log^2 n \log \Delta$-approximation algorithm, 
where $\Delta$ is the maximum degree in $G$.
We also consider the \ProblemStar{} problem, show that it is inapproximable within
ratio of $c \log n$ for some $c > 0$ and show how to reduce it to the \ProblemSetCover{}
problem to obtain a $O(\log n)$ approximation.
Finally, we consider the \ProblemPath{} problem and show that it is inapproximable at all.

\section{\Problem{}}
Recall that instance of \ProbGroup{} is a tuple $(G, w, \mathcal{G})$,
where $G = (V, E)$ is an undirected graph, $w:E \to \mathcal{R}$ is a weight function, 
and $\mathcal{G} \subseteq 2^V$ is a family of groups of vertices.
A Steiner group tree of is a sub-tree of $G$, $T = (U, F)$
such that $g \cap U \neq \emptyset$ for each $g \in \mathcal{G}$.
The cost of such tree is $\sum_e \in F w(e)$.
Given an instance of \ProbGroup{} we are looking for a minimum cost Steiner group tree
of $G$.

In this section we show that \Prob{} is equivalent to \ProbGroup{} from 
approximability perspective, 
that is every approximation algorithm to one problem yields the same approximation 
ratio for the other problem.
To show this we introduce approximation preserving reductions from \Prob{} to \ProbGroup{}
and vice-versa. 

	\subsection{\Prob{} $\leq_p$ \ProbGroup{}}
We start by showing an approximation-preserving reduction from \Prob{} to
\ProbGroup{}.
Given an instance of \Prob{}, $(G, w)$, where $G = (V, E)$, 
we define an instance of \ProbGroup{},
$(G, w, \mathcal{G})$.
We now have to define $\mathcal{G}$: for each vertex $v \in V$ we define 
$g_v \in \mathcal{G}$ to be $\{v\} \cup N(v)$.
Figure~\ref{fig:prob-leq-group} depicts this transformation.
Note that the number of groups is $n$ and that the size of the largest group is $\Delta$.
Clearly this transformation can be done in polynomial time.
The following two claims show that this is an approximation preserving reduction. 

\begin{figure}
\begin{center}
\begin{tikzpicture}[]
\foreach[count=\i] \x \y in {
	0/0
	,1/2,1/-1
	,2/0
	,3/1,3/-2
	,4/0
}{
	\node(\i) at(\x,\y) {\i};
}

\foreach \u \v in {
	1/2,1/3%
	,2/4,2/5%
	,3/4,3/6,3/7%
	,4/5,4/7%
	,5/7%
	,6/7%
}{
	\draw (\u) -- (\v);
}

\begin{scope}[xshift=6cm]
\foreach[count=\i] \x \y in {
	0/0
	,1/2,1/-1
	,2/0
	,3/1,3/-2
	,4/0
}{
	\node(\i) at(\x,\y) {\i};
}

\foreach \u \v in {
	1/2,1/3%
	,2/4,2/5%
	,3/4,3/6,3/7%
	,4/5,4/7%
	,5/7%
	,6/7%
}{
	\draw (\u) -- (\v);
}
\end{scope}

\draw[dashed, blue]
(2.north) to[out=0,in=90]
(5.east) to[out=-90,in=0]
(4.south) to[out=180,in=-90]
(1.west) to[out=90,in=180]
(2.north)
;

\end{tikzpicture}
\end{center}
\caption{\label{fig:prob-leq-group}
From left to right:
a) A \Prob{} instance (weights are omitted).
b) A corresponding \ProbGroup{} instance on the same weighted graph.
For each vertex $v$ we define a group
that contains its neighborhood to ensure that in any group Steiner tree there is at least
one vertex that dominate $v$.
$g_2$ is marked in the figure with a dashed blue line.  
}
\end{figure}
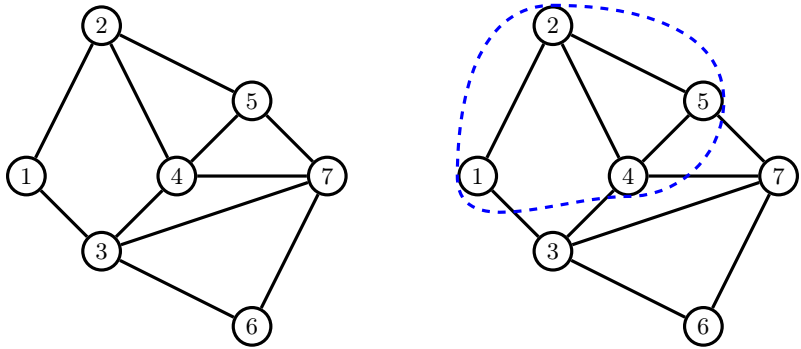

\begin{claim}
Any dominating tree, $T$, in $(G, w)$ 
is a feasible group Steiner tree in $(G, w, \mathcal{G})$.
\end{claim}

\begin{proof}
Assume for contradiction that $T$ is not feasible group Steiner tree, that is, there is 
a group $g_v$ such that none of the vertices in $g_v$ is spanned by $T$, that is $v$
is not in $T$ nor any of its neighbors and thus $T$ is not a dominating tree - contradiction. 
\end{proof}
 
\begin{claim}
Any feasible group Steiner tree in $(G, w, \mathcal{G})$, $T$, is a dominating tree
in $(G, w)$.
\end{claim}

\begin{proof}
Assume for contradiction that $T$ is not a dominating tree, that is, there is 
a vertex $v$ not in $T$ such that none of its neighbors belong to $T$, thus, 
$T$ does not cover $g_v$ - contradiction. 
\end{proof}
	\subsection{\Prob{} $\geq_p$ \ProbGroup{}}
We now show an approximation-preserving reduction from \ProbGroup{} to
\Prob{}.
Given an instance of \ProbGroup{}, $(G, w, \mathcal{G})$, where $G = (V, E)$, 
we define an instance of \Prob{},
$(G', w')$ where $G' = (V', E')$, and:
\begin{description}
\item[V'] - $V \cup \{g_i : g_i \in \mathcal{G}\}$
\item[E'] - $E \cup \{vg_i : v \in g_i\} \cup (V \times V) \setminus E$
\item[w'] - $
w'(e) = 
\begin{cases}
w(e) 	& e \in E
\\
\infty & \text{otherwise}
\end{cases}
$
\end{description}
Figure~\ref{fig:prob-geq-group} depicts this transformation.
Clearly the above transformation can be done in polynomial time,
The following two claims show that this is an approximation preserving reduction:

\begin{figure}
\begin{center}
\begin{tikzpicture}[]
\foreach[count=\i] \a in {0,72,...,288}{
	\node(\i) at(\a:2) {\i};
}

\foreach \u \v in {
	1/2,2/3,3/4,4/1,4/5,5/1%
}{
	\draw (\u) -- (\v);
}

\draw[dashed, blue]
(2.north) to[out=0,in=0]
(1.south) to[out=180,in=180]
(2.north)
;

\draw[dotted, green]
(1.north) to[out=0,in=0]
(5.south) to[out=180,in=180]
(1.north)
;

\draw[dash dot, orange]
(3.north) to[out=0,in=0]
(4.south) to[out=180,in=180]
(3.north)
;

\begin{scope}[xshift=65mm]
\foreach[count=\i] \a in {0,72,...,288}{
	\node(\i) at(\a:2) {\i};
}

\foreach \u \v in {
	1/2,2/3,3/4,4/1,4/5,5/1%
}{
	\draw (\u) -- (\v);
}

\node(g1) at(31:3){$g_1$};
\node(g2) at(-31:3){$g_2$};
\node(g3) at(180:3){$g_3$};

\foreach \u \v in {
	1/3%
	,2/4,2/5%
	,3/5%
	,g1/1,g1/2%
	,g2/1,g2/5%
	,g3/3,g3/4%
}{
	\draw[dashed] (\u) -- (\v);
}
\end{scope}

\end{tikzpicture}
\end{center}
\caption{\label{fig:prob-geq-group}
From left to right:
a) An instance of \ProbGroup{} (weights are omitted), 
groups are marked by a dashed blue, dotted green, and dashed-dotted lines respectively.  
b) A corresponding \Prob{} instance: we add a vertex for each group and connect it 
to all terminals in the group.
Weights for the original edges remain intact, dashed edges have infinity weight.   
}
\end{figure}
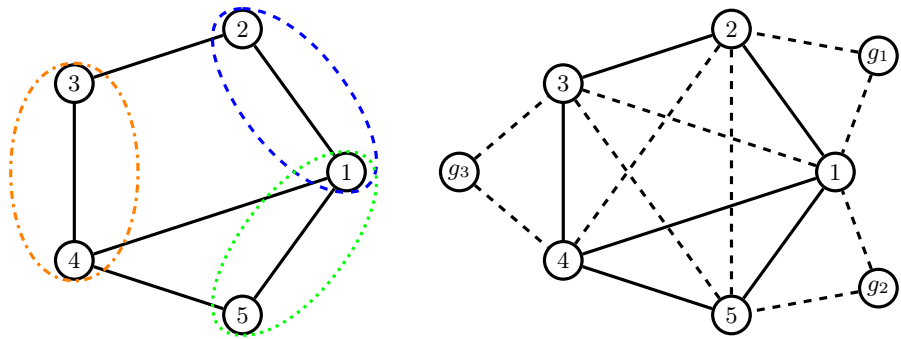

\begin{claim}
Any dominating tree with finite weight, $T$, in $(G', w')$ is a feasible group Steiner tree in 
$(G, w, \mathcal{G})$.
\end{claim}

\begin{proof}
Observe first that $T$ contains only original edges or otherwise its weight is infinite.  
Assume for contradiction that $T$ is not feasible group Steiner tree, that is, there is 
a group $g_i$ such that none of the vertices in $g_i$ is spanned by $T$, that is $g_i$
is not in dominated by any vertex in $T$ - contradiction. 
\end{proof}
 
\begin{claim}
Any feasible group Steiner tree in $(G, w, \mathcal{G})$, $T$, is a dominating tree
in $(G', w')$.
\end{claim}

\begin{proof}
First, observe that any original vertex dominate all other original vertices.
Assume for contradiction that $T$ is not a dominating tree, that is, there is 
a new vertex $g_i$ that is not dominated by $T$ which implies that 
$T$ does not cover $g_i$ - contradiction. 
\end{proof}

\section{\ProblemStar{}}
In this section we consider the \ProblemStar{} problem, 
that is the \Problem{} problem when restricted to stars, i.e.
given an undirected weighted graph $(G, w)$ find a 
minimum weight dominating sub-star (a tree with diameter at most 2).
We start by showing that this problem cannot be approximated 
within $c\log n$ for some $c > 0$ unless $P = NP$.
Then we show how to reduce the problem to a set cover instance to achieve a 
$O(\log n)$-approximation algorithm.  
	\subsection{Hardness}
We show an approximation preserving reduction from the \ProblemDom{} problem (\ProbDom)
to \ProblemStar{}.
Given an (unweighted) instance of \ProbDom{} $G = (V, E)$ we create an instance of \ProbStar{}
$(G', w)$ where $G' = (L \cup R \cup {c}, E' \cup \{c\} \times L$).
\begin{description}
\item[L] - $\{v_l : v \in V\}$
\item[R] - $\{v_r : v \in V\}$
\item[E'] - $\{u_lv_r : uv \in E\}$
\end{description}
We also set $w(e) = \infty$ for every $e \in E'$ and $w(cv_l) = 1$ for every $v_l \in L$.
Figure~\ref{fig:star-hardness} depicts this transformation.
Clearly the above transformation can be done in polynomial time,
The following two claims show that this is an approximation preserving reduction:

\begin{figure}
\begin{center}
\begin{tikzpicture}[]
\foreach[count=\i] \a in {0,72,...,288}{
	\node(\i) at(\a:2) {\i};
}

\foreach \u \v in {
	1/2,2/3,3/4,4/1,4/5,3/1%
}{
	\draw (\u) -- (\v);
}

\node(c) at(4,0) {c};

\begin{scope}[xshift=6cm, yshift=-3cm]
\foreach \i in {1,...,5}{
	\node(l\i) at(0,\i) {\i};
	\node(r\i) at(2,\i) {\i};
}

\foreach \u \v in {
	1/2,2/3,3/4,4/1,4/5,3/1%
}{
	\draw (l\u) -- (r\v);
	\draw (r\u) -- (l\v);
}

\foreach \i in {1,...,5}{
	\draw[dashed] (c) -- (l\i);
}
\end{scope}
\end{tikzpicture}
\end{center}
\caption{\label{fig:star-hardness}
From left to right:
a) An instance of \ProbDom{} (unweighted graph)
b) The corresponding instance of \ProbStar{}, 
the weight on the original edges (solid black) is infinity 
and the weight on the new edges is 1. 
}
\end{figure}
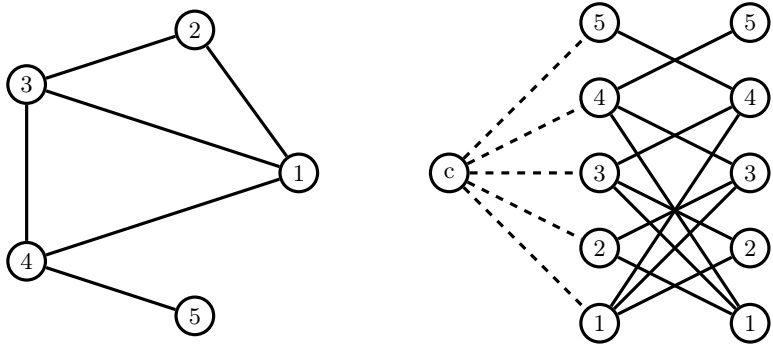 

\begin{claim}
If $D$ is a dominating set in $G$ then 
$S = (\{c\} \cup \{v_l : v \in D\}, \{cv_l: v \in D\})$ 
is a dominating star in $G'$ that weigh $|D|$. 
\end{claim}

\begin{proof}
$S$ weights $D$ by the definition of $G'$ and $S'$.
Assume for contradiction that it is not a dominating star and let $v$ be a non dominated 
vertex then $v$ is also not dominated in $G$ under $D$ - contradiction.
\end{proof}

\begin{claim}
If $S = (U, F)$ is a dominating star in $G'$ of weight $k < \infty$ 
then $\{v : v_l \in U \setminus \{c\}\}$ is a dominating set in $G$ of size $k$.
\end{claim}

\begin{proof}
Observe that any star that contains more than 3 vertices must be centered at $c$ or otherwise
its weight is infinity.
Thus the leaf of the star dominate all vertices in $R$, by construction, this mean that the
corresponding vertices in the original graph dominate all other vertices.
\end{proof}
	\subsection{$\log n$-Approximation}
We now show how to reduce \ProbStar{} to an instance of the \ProblemSetCover{} problem
(\ProbSetCover{}) in order to obtain a $O(\log n)$-approximation algorithm.
This is the best one can hope for if $P \neq NP$.

Without loss of generality, we assume that the center of the dominating star is known,
or otherwise we can solve the problem for every vertex in the graph assuming it is 
the center.  
Given an undirected weighted graph $(G, w)$ and a the center of the dominating star
$c \in V$ we create the following an instance of \ProbSetCover{} 
$(U, \mathcal{S}, w')$ as follow:
\begin{description}
\item[$U$] - $V \setminus N(c)$
\item[$\mathcal{S}$] - $\{S_v : cv \in E\}$
\item[$S_v$] - $N(v) \cap U$
\item[$w'$] - $w'(S_v) = w(cv)$ 
\end{description}
Clearly the above transformation can be done in polynomial time.
The following two claims show that this is an approximation preserving reduction:

\begin{claim}
If $S = (c, L)$ is a dominating star in $G$ then 
$\mathcal{C} = \{S_v : v \in L\}$ is a set cover in $(U, \mathcal{S}, w')$,
moreover, $w(S) = w'(\mathcal{C})$.  
\end{claim}

\begin{proof}
$w(S) = w'(\mathcal{C})$ by construction.
Now, assume for contradiction that $\mathcal{C}$ is not a set cover and let $v$
be an uncover element then $v$ is also not dominated by $S$ - contradiction.
\end{proof}

\begin{claim}
If $\mathcal{C}$ is a set cover in $(U, \mathcal{S}, w')$ 
then $S = (c, \{v : S_v \in \mathcal{C}\})$ is a dominating star in $G$,
moreover, $w(S) = w'(\mathcal{C})$.  
\end{claim}

\begin{proof}
$w(S) = w'(\mathcal{C})$ by construction.
Now, assume for contradiction that $S$ is not a dominating star and let $v$
be an undominated vertex then $v$ is also uncovered by $\mathcal{C}$ - contradiction.
\end{proof}

\section{\ProblemPath{}}
We show that the \ProblemPath{} problem (\ProbPath{}) cannot be approximated at all
unless $P = NP$.
We show a reduction from the \ProblemHam{} problem (\ProbHam{}).
In \ProbHam{} we are given an undirected graph $G = (V, E)$ and we are 
asked to decide if there is an Hamiltonian path 
(a simple path traversing all the vertices in $V$) in $G$ or not.
\ProbHam{} is one of the classical NP-hard problems.

Given an instance of \ProbHam{}, $G = (V, E)$, we define an instance of \ProbPath{},
$(G', w)$, where $G' = (V \cup \{v' : v \in V\}, E \cup \{vv' : v \in V\})$.
We set $w(e) = 0$ for every edge $e \in E$ and set $w(e) = \infty$ otherwise.
Figure~\ref{fig:hamiltonian} depicts this transformation.  
We now claim that any (multiplicative) approximation algorithm for \ProbPath{}
can solve \ProbHam{}.
Let $G$ be an instance of (decision problem) \ProbHam{}, 
and denote by $A(G', w)$ the value of (approximation) algorithm, $A$, 
on the corresponded \ProbPath{} instance, then:

\begin{claim}
$A(G', w) = 0 \iff G \in \ProbHam{}$. 
\end{claim} 

\begin{proof}
Let $P$ be a dominating path in $(G', w)$ with value 0, 
then it uses only edges of $E$, moreover $P$ is an hamiltonian path in $G$ or otherwise
there is a vertex $v'$ that is not dominated by $P$.
Now, let $P$ be an hamiltonian path in $G$ then $P$ is a dominating path in $G'$ with value
0, thus, any (multiplicative) approximation algorithm must also find a path value 0.
\end{proof}

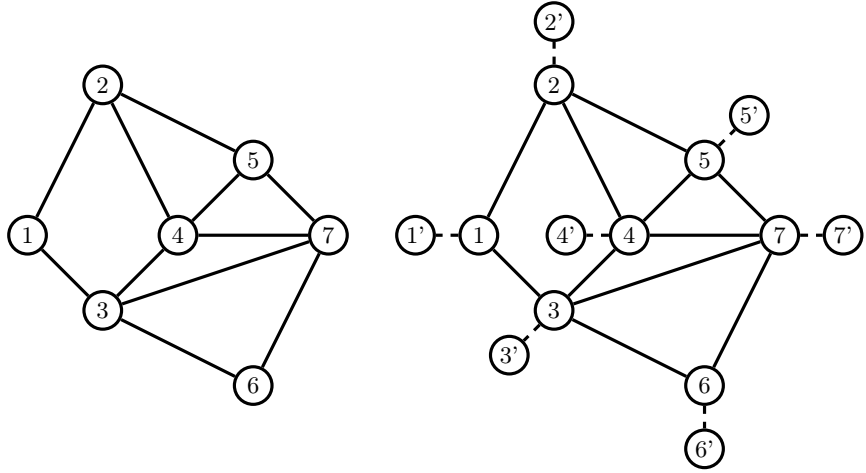
\begin{figure}
\begin{center}
\begin{tikzpicture}[]
\foreach[count=\i] \x \y in {
	0/0
	,1/2,1/-1
	,2/0
	,3/1,3/-2
	,4/0
}{
	\node(\i) at(\x,\y) {\i};
}

\foreach \u \v in {
	1/2,1/3%
	,2/4,2/5%
	,3/4,3/6,3/7%
	,4/5,4/7%
	,5/7%
	,6/7%
}{
	\draw (\u) -- (\v);
}

\begin{scope}[xshift=6cm]
\foreach[count=\i] \x \y in {
	0/0
	,1/2,1/-1
	,2/0
	,3/1,3/-2
	,4/0
}{
	\node(\i) at(\x,\y) {\i};
}

\foreach \u \v in {
	1/2,1/3%
	,2/4,2/5%
	,3/4,3/6,3/7%
	,4/5,4/7%
	,5/7%
	,6/7%
}{
	\draw (\u) -- (\v);
}

\foreach[count=\i] \p in {
	left,above,below left,left,above right,below,right%
}{
	\node(\i')[\p=3mm of \i] {\i'};
	\draw[dashed] (\i) -- (\i');
}
\end{scope}
\end{tikzpicture}
\end{center}
\caption{\label{fig:hamiltonian}
From left to right
a) An instance of the \ProblemHam{} problem.
b) The corresponded \ProblemPath{} instance, 
original edges have zero weight, dashed edges have infinite weight.
}
\end{figure} 

\section{Conclusion}
Our main result shows that on general graphs any approximation algorithm for
the \Problem{} problem yields the same approximation result for the 
\ProblemGroup{} problem and vice versa.
This result might give another perspective and maybe shed some light on
the \ProblemGroup{} problem.

We remark, however, that the two problems are not equivalent.
A good example is when the input graph is a tree, \ProbGroup{} is known to 
be as hard to approximate as \ProbSetCover{} even in this case while \Prob{} is 
trivially solvable on trees.
Thus, there is also a place to studying each of the problems on its own for particular 
families of graphs.

\bibliographystyle{plain}
\bibliography{main}

\end{document}